\documentclass[acmsmall]{acmart}

\usepackage[show]{chato-notes}  

\usepackage{lipsum}
\usepackage{array}
\usepackage{paralist} 
\usepackage{subfig}
\usepackage{makecell,amsmath}
\usepackage{forest} 
\usepackage[ruled,lined,linesnumbered]{algorithm2e}
\usepackage{hyperref}
\usepackage{xcolor}
\usepackage{color, colortbl}
\usepackage{float}
\usepackage{graphicx}
\usepackage[section]{placeins}
\hypersetup{
	colorlinks=false,
	linkcolor={red!20!black},
	citecolor={green!20!black},
	urlcolor={blue!20!black}
}

\AtBeginDocument{%
  \providecommand\BibTeX{{%
    \normalfont B\kern-0.5em{\scshape i\kern-0.25em b}\kern-0.8em\TeX}}}

\setcopyright{iw3c2w3}
\copyrightyear{2020}
\newcolumntype{L}[1]{>{\raggedright\let\newline\\\arraybackslash\hspace{0pt}}m{#1}}
\newcolumntype{C}[1]{>{\centering\let\newline\\\arraybackslash\hspace{0pt}}m{#1}}
\newcolumntype{R}[1]{>{\raggedleft\let\newline\\\arraybackslash\hspace{0pt}}m{#1}}

\newcommand{\alphaadj}[0]{\ensuremath{\alpha_\textit{corr}}}

\newcommand{\minv}{\ensuremath{m^{-1}}}

\newcommand{\mtable}{\ensuremath{\operatorname{mTable}}}

\newcommand{\failprob}{\ensuremath{P_{\operatorname{fail}}}}
\newcommand{\successprob}{\ensuremath{P_{\operatorname{succ}}}}

\newcommand{\algoRecursive}[0]{{\sc SuccessProbability}\xspace}
\newcommand{\algoBinomBinary}[0]{{\sc AlphaAdjustment}\xspace}

\newcommand{\algoMtable}[0]{{\sc ConstructMTable}\xspace}

\SetCommentSty{mycommfont}
\SetKwInOut{AlgInput}{input}
\SetKwInOut{AlgOutput}{output}
\SetKwComment{AlgComment}{// }{}
\SetAlCapFnt{\small}
\SetAlCapNameFnt{\footnotesize}
\newcommand{\nosemic}{\renewcommand{\@endalgocfline}{\relax}}
\newcommand{\dosemic}{\renewcommand{\@endalgocfline}{\algocf@endline}}
\newcommand{\pushline}{\Indp}
\newcommand{\popline}{\Indm}

\begin{document}


\title{A Note on the Significance Adjustment for FA*IR with Two Protected Groups}

\author{Meike Zehlike}
\affiliation{
	\institution{Humboldt Universit\"at zu Berlin}
}
\affiliation{
	\institution{Max-Planck-Inst. for Software Systems}
	\streetaddress{Campus E1 5}
	\postcode{66123}
	\city{Saarbr\"ucken}
	\country{Germany}
}
\email{meikezehlike@mpi-sws.org}

\author{Tom S\"uhr}
\affiliation{
	\institution{Technische Universit\"at Berlin}
\country{Germany}}
\email{tom.suehr@googlemail.com}

\author{Carlos Castillo}
\affiliation{
	\institution{Universitat Pompeu Fabra}
\city{Barcelona}
\country{Spain}}
\email{chato@acm.org}
\renewcommand{\shortauthors}{Zehlike et al.}

\begin{abstract}
\textbf{Abstract: }In this report we provide an improvement of the significance adjustment from the FA*IR algorithm in \citet{zehlike2017fair}, which did not work for very short rankings in combination with a low minimum proportion $p$ for the protected group.
We show how the minimum number of protected candidates per ranking position can be calculated exactly and provide a mapping from the continuous space of significance levels ($\alpha$) to a discrete space of tables, which allows us to find $\alphaadj$ using a binary search heuristic.

\end{abstract}




\maketitle

In this report we describe a correction of the significance adjustment procedure from~\cite{zehlike2017fair}, which did not work for very small $k$ and $\alpha$.

For binomial distributions, i.e. where only one protected and one non-protected group is present, the inverse CDF can be stored as a simple table, which we compute using Algorithm~\ref{alg:constructMTable}.
We will call such a table \textit{mTable}.
\vspace{-2mm}
\begin{algorithm}[h]
	\caption{Algorithm \algoMtable computes the data structure to efficiently verify or construct a ranking that satisfies binomial ranked group fairness.}
	\label{alg:constructMTable}
	\small
	\AlgInput{$k$, the size of the ranking to produce; $p$, the expected proportion of protected elements; $\alphaadj$, the significance for each individual test.}
	\AlgOutput{$ \mtable $: A list that contains the minimum number of protected candidates required at each position of a ranking of size $k$.}
	$\mtable \leftarrow [k]$ \AlgComment{list of size $k$}
	\For{$i \leftarrow 1$ \KwTo $k$}{
		$\mtable[i] \leftarrow F^{-1}(i,p,\alphaadj)$ \AlgComment{the inverse binomial cdf}
	}
	\Return{$ \mtable $ }
\end{algorithm}
\vspace{-2mm}

Table~\ref{tbl:ranked_group_fairness_table} shows an example of mTables for different $ k $ and $ p $, using $\alpha=0.1$.
For instance, for $p=0.5$ we see that at least 1 candidate from the protected group is needed in the top 4 positions, and 2 protected candidates in the top 7 positions.
\begin{table}[ht!]
	\small\begin{tabular}{r|cccccccccccc}
		\diaghead{some text}%
		{p}{k}&
		1 & 2 & 3 & 4 & 5 & 6 & 7 & 8 & 9 & 10 & 11 & 12 \\ \midrule
		0.1      & 0 & 0 & 0 & 0 & 0 & 0 & 0 & 0 & 0 & 0  &  0 &  0 \\
		0.3      & 0 & 0 & 0 & 0 & 0 & 0 & 1 & 1 & 1 & 1  &  1 &  2 \\
		0.5      & 0 & 0 & 0 & 1 & 1 & 1 & 2 & 2 & 3 & 3  &  3 &  4 \\
		0.7      & 0 & 1 & 1 & 2 & 2 & 3 & 3 & 4 & 5 & 5  &  6 &  6 \\
		\bottomrule
	\end{tabular}
	\caption{Example values of $m_{\alpha,p}(k)$, the minimum number of candidates in the protected group that must appear in the top $k$ positions to pass the ranked group fairness criteria with $\alpha=0.1$ in a binomial setting.}
	\label{tbl:ranked_group_fairness_table}
\end{table}

Figure~\ref{fig:why-adjustment-is-needed-binomial} shows that we need a correction for $\alpha$ as we are testing multiple hypothesis in the ranked group fairness test, namely $k$ of them (note that the scale is logarithmic).
In the following, we show that the special case of having only one protected group offers possibilities for verifying ranked group fairness efficiently.
A key advantage of considering just one protected group, rather than multiple, is that we can calculate the exact failure probability $\failprob$ (i.e. a fair ranking gets rejected by the ranked group fairness test), which results in an efficient binary search for $\alphaadj$.
%

First we introduce the necessary notation for the binomial case and describe how we calculate the exact $\failprob$.
Then we show that we can divide the continuum of possible $\alpha$ values in discrete parts in order to be able to apply efficient binary search for the most accurate $\alphaadj$.
Last we analyze the complexity of the proposed algorithms.
%
\begin{figure}[t!]
	\centering
	{\includegraphics[width=.48\textwidth]{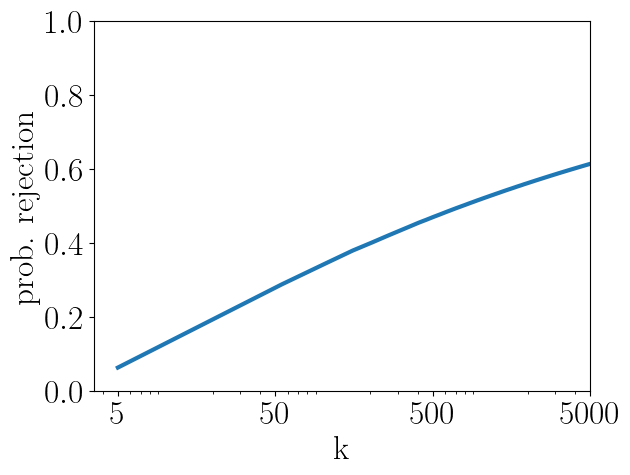}}
	\caption{
		Probability that a fair ranking created by a Bernoulli process with $p=0.5$ fails the ranked group fairness test.\label{fig:why-adjustment-is-needed-binomial}
		Experiments on data generated by a simulation, showing the need for multiple tests correction.
		The data has one protected group with a ranking created by  a Bernoulli process (Fig.~\ref{fig:why-adjustment-is-needed-binomial}).
		Rankings should have been rejected as unfair at a rate $\alpha = 0.1$.
		However, we see that the rejection probability increases with $k$.
		Note the scale of $k$ is logarithmic.}
	\label{fig:need-for-model-adjustment}
\end{figure}

\section{Success Probability for One Protected Group}\label{subsubsec:adjustment-binomial}

The probability $\successprob$ that a ranking created following the procedure shown by~\citet{yang2016measuring} passes the ranked group fairness test with parameters $p$ and $\alpha$ can be computed using the following procedure:
Let $m(k) = m_{\alpha,p}(k) = F^{-1}(k,p,\alpha)$ be the number of protected elements required up to position $k$.
Let $\minv(i) = k$ s.t. $m(k) = i$ be the position at which $i$ or more protected elements are required.
Let $b(i) = \minv(i) - \minv(i-1)$ (with $\minv(0) = 0$) be the size of a ``block,'' that is, the gap between one increase and the next in $m(\cdot)$.
We call the $k$-dimensional vector $(m(1), m(2), \ldots , m(k))$ a \emph{mTable}.
An example is shown on Table~\ref{tbl:05:example_blocks}.
\begin{table}[ht!]
	\centering
	\begin{tabular}{cccccccccccccc}\toprule
		$k$    & 1 & 2 & 3 & \textbf{{4}} & 5 & 6 & \textbf{7} & 8 & \textbf{9} & 10 & 11 & \textbf{12} \\
		\midrule
		$m(k)$ & 0 & 0 & 0 & \multicolumn{1}{c|}{1} & 1 & 1 & \multicolumn{1}{c|}{2} & 2 & \multicolumn{1}{c|}{3} & 3  & 3  & \multicolumn{1}{c}{4}\\
		Inverse   & \multicolumn{4}{c|}{$\minv(1)=4$}
		& \multicolumn{3}{c|}{$\minv(2)=7$}
		& \multicolumn{2}{c|}{$\minv(3)=9$}
		& \multicolumn{3}{c}{$\minv(4)=12$}\\
		Blocks       & \multicolumn{4}{c|}{$b(1)=4$}
		& \multicolumn{3}{c|}{$b(2)=3$}
		& \multicolumn{2}{c|}{$b(3)=2$}
		& \multicolumn{3}{c}{$b(4)=3$}\\
		\bottomrule
	\end{tabular}
	\caption[Example of different block sizes]{Example of $m(\cdot)$, $\minv(\cdot)$, and $b(\cdot)$ for $p=0.5, \alpha=0.1$.}
	\label{tbl:05:example_blocks}
\end{table}

\noindent Furthermore let
\begin{equation}
	\label{eq:05:combinations}
	I_{m(k)} = \{ v = (i_1, i_2, \ldots, i_{m(k)}): \forall \ell' \in \lbrace 1,\ldots,m(k) -1 \rbrace, 0 \le i_{\ell'} \le b(\ell') \wedge \sum_{j=1}^{\ell'} i_j \ge \ell' \}
\end{equation}
represent all possible ways in which a fair ranking\footnote{Note that we do not consider rankings of size 0, which always pass the test.} generated by the method of \citet{yang2016measuring} can pass the ranked group fairness test, with $i_j$ corresponding to the number of protected elements in block $j \; (\text{with } 1 \le j \le k)$.
As an example consider again Table~\ref{tbl:05:example_blocks}: the first block contains four positions, i.e. $b(1)=4$ and this block passes the ranked group fairness test, if it contains at least one protected candidate, hence $i_1 \in \{1, 2, 3, 4\}$.
The probability of considering a ranking of $k$ elements (i.e. $m(k)$ blocks) unfair, is:
\begin{equation}
	\label{eq:05:failureProb}
	\failprob = 1 - \successprob = 1 - \sum_{v \in I_{m(k)}} \prod_{j=1}^{m(k)} f(v_j; b(j), p)
\end{equation}
\noindent where $f(x;b(j),p) = Pr(X = x)$ is the probability density function (PDF) of a binomially distributed variable $X \sim Bin(b(j), p)$.
However, if calculated naively this expression is intractable because of the large number of combinations in $I_{m(k)}$.

\setlength{\textfloatsep}{2pt}
\begin{algorithm}[t!]
	\caption{Algorithm \algoRecursive computes the probability, that a given mTable accepts a fair ranking (see right term of Eq.~\ref{eq:05:failureProb}).}
	\label{alg:05:successProb} 
	\small
	\AlgInput{
		$\texttt{b[]}$ list of block lengths (Table~\ref{tbl:05:example_blocks}, 3rd line);\\ 
		$\texttt{maxProtected}$ the sum of all entries of $\texttt{b[]}$;\\
		$\texttt{currentBlockIndex}$ index of the current block; \\
		$\texttt{candidatesAssigned}$ number of protected candidates assigned for the current possible solution; \\
		$p$, the expected proportion of protected elements.}
	\AlgOutput{The probability of accepting a fair ranking.}
	
	\If{$\texttt{b[].length} = 0$}{
		\Return{$1$}
	}
	\tcp{we need to assign at least one protected candidate to each block}
	$\texttt{minNeededThisBlock} \leftarrow \texttt{currentBlockIndex} - \texttt{candidatesAssigned}$\\
	\tcp{if we already assigned enough candidates, minNeededThisBlock = 0 (termination condition for the recursion)}
	\If{$\texttt{minNeededThisBlock} < 0$}{
		$\texttt{minNeededThisBlock} \leftarrow 0$
	}
	$\texttt{maxPossibleThisBlock} \leftarrow \textit{argmin}(\texttt{b[0]}, \texttt{maxProtected})$ \\
	$\texttt{assignments} \leftarrow 0$ \\
	$\texttt{successProb} \leftarrow 0$ \\
	\tcp{sublist without the first entry of $\texttt{b[]}$}
	$\texttt{b\_new[]} \leftarrow \textit{sublist}(\texttt{b[]}, 1, \texttt{b[].length})$ \label{algoline:05:suffixes}\\
	$\texttt{itemsThisBlock} \leftarrow \texttt{minNeededThisBlock}$\\
	\While{$\texttt{itemsThisBlock} \leq \texttt{maxPossibleThisBlock}$}{
		$\texttt{remainingCandidates} \leftarrow \texttt{maxProtected} - \texttt{itemsThisBlock}$ \\
		$\texttt{candidatesAssigned} \leftarrow \texttt{candidatesAssigned} + \texttt{itemsThisBlock}$ \\
		\tcp{each recursion returns the success probability of \emph{all possible ways} to fairly rank protected candidates after this block}
		$\texttt{suffixSuccessProb} \leftarrow \textsc{\algoRecursive} ( $ \\ \pushline $ \texttt{remainingCandidates},\texttt{b\_new[]}, \texttt{currentBlockIndex} + 1,$ \\ $ \texttt{candidatesAssigned})$ 
		\label{algoline:05:recursion}\\
		\popline $\texttt{totalSuccessProb} \leftarrow \texttt{totalSuccessProb} \; + $ \\ \pushline $ \textsc{PDF}(\texttt{maxPossibleThisBlock}, \texttt{itemsThisBlock}, p) \; \cdot $ \\ $ \texttt{suffixSuccessProb}$ \label{algoline:05:pdf}\\
		\popline $\texttt{itemsThisBlock} \leftarrow \texttt{itemsThisBlock} + 1$\\
	}
	\Return{probability of accepting a fair ranking: $\texttt{totalSuccessProb}$ }
\end{algorithm}
We therefore propose a dynamic programming method, Algorithm~\ref{alg:05:successProb}, which computes the probability that a fair ranking passes the ranked group fairness test (i.e. the right term of Equation~\ref{eq:05:failureProb}) \emph{recursively}.
Note that because of the combinatorial complexity of the problem a \emph{simple closed-form expression} to compute $\failprob$ is unlikely to exist.
%
%
%
The algorithm breaks the vector $v = (i_1, i_2, \ldots ,i_{\ell})$ of Equation~\ref{eq:05:combinations} into a \textit{prefix} and a \textit{suffix} (Alg.~\ref{alg:05:successProb}, Line~\ref{algoline:05:suffixes}).
We call $i_1$ the \textit{prefix} of $(i_2, \ldots, i_{\ell})$, and $(i_2, \ldots, i_{\ell})$ the \textit{suffix} of $i_1$.
The algorithm starts with a prefix and calculates all possible suffixes, that pass the ranked group fairness test, recursively (Line~\ref{algoline:05:recursion}).
Consider the following example: for the first prefix $i_1 = 1$ the algorithm computes all possible suffixes, where we rank exactly one protected candidate in the first block.
For this prefix $i_1$, combined with each possible suffix $v \setminus i_1$, we calculate the success probability $\prod_{j=1}^{m(k)} f(v_j; b(j), p)$ for each instance of $v$ (Line~\ref{algoline:05:pdf}).
In the next recursion level we start with a new prefix, let us say $i_1=1, i_2 =1$.
The algorithm computes all possible suffixes, i.e. all rankings where we rank exactly one protected candidates in the first block and one protected candidate in the second block.
Then it computes the respective success probabilities.
This procedure continues for $m(k)$ iterations.
After that the whole program starts again with $i_1=2$ and is repeated until the maximum number of protected candidates is reached, in our case $b(1)=4$.
All intermediate success probabilities are added up (Line~\ref{algoline:05:pdf}) to the total success probability (see Eq.~\ref{eq:05:failureProb}) of the mTable that was created given $k, p, \alpha$.

Note that there are at most $\prod_{j=1}^{m(k)}b(j)$ possible combinations to distribute the protected candidates within the blocks.
Furthermore many $v \in I_{m(k)}$ share the same prefix and hence have the same probability density value for these prefixes.
To reduce computation time the algorithm stores the binomial probability density value for each prefix in a hash map with the prefix as key and the respective pdf as value.
Thus the overall computational complexity becomes $O(\prod_{j=1}^{m(k)}b(j) \cdot O(\texttt{binomPDF}))$.

\section{Finding the Correct mTable}\label{subsec:finding-mtable}
We call an mTable \textit{correct} if it has an overall success probability of $\successprob = 1-\alpha$.
However, given parameters $k,p,\alpha$, $\successprob$ will be greater or equal to $1-\alpha$. Thus, we need a corrected $\alphaadj \leq \alpha$ in order to compute the correct mTable. Unfortunately there is no way to compute $\alphaadj$ directly, which is why we have to search for the correct mTable, hence $\alphaadj$.
We propose Algorithm~\ref{alg:05:binarySearch} that takes  parameters $k, p, \alphaadj$ as input and returns the correct mTable and $\alphaadj$.
%
%
It sequentially creates mTables (recall that these are $k$-dimensional vectors of the form $(m(1), m(2), \ldots , m(k))$) for different values of $\alphaadj$, and then calls Algorithm~\ref{alg:05:successProb} to calculate their success probability until it finds the correct mTable with overall failure probability $\failprob = \alpha$ .
Our goal is to use binary search to select possible candidates for $\alphaadj$ systematically.

However, to be able to do binary search, we need a discrete measure for the $\alpha$-space to search on, otherwise the search would never stop. Specifically, we could never be sure if we found the mTable with the minimum difference of $\successprob$ to $1-\alpha$. A binary search would further and further divide an interval between two $\alpha$ values. The only chance to verify that we do not have to search further is by comparing the resulting mTables of different $\alpha$ values. We will see, that we can do that by comparing the sum of the entries in the mTable and that there exist only a limited number of them.
Furthermore, to reduce complexity we only want to consider mTables with certain properties, which we define in the following paragraph.
A $k$-dimensional vector (e.g. $(0,0,1,2,3)$) has to have two properties in order to constitute a mTable, rather than just a vector of natural numbers: it has to be \emph{valid} and \emph{legal}.
\begin{definition}[Valid mTable]
	\label{def:05:valid-mtable}
	The $\text{mTable}_{p,k,\alpha}=(m(1) , m(2) , \ldots , m(k))$ is \emph{valid} if and only if, $m(i) \leq m(j)$ for all $i,j \in \lbrace 0, \ldots, k \rbrace$ with $i < j$ and $m(i)=n \Rightarrow m(i+1) \leq n+1$.
\end{definition}
\noindent It is easy to see that many valid mTables exist.
They correspond to all $k$-dimensional arrays with integers monotonically increasing by array indices.
However we only want to consider those valid mTables for our ranked group fairness test that have been created by the statistical process in~\citet{yang2016measuring}.
We call these \textit{legal mTables}.
\begin{definition}[Legal mTable]
	\label{def:05:legal-mtable}
	A $\text{mTable}_{p,\alpha,k}$ is \textit{legal} if and only if there exists a $p,k,\alpha$ such that
	$\texttt{constructMTable}(p,k,\alpha)=\text{mTable}_{p,\alpha,k}$.
\end{definition}
Definition \ref{def:05:legal-mtable} restricts the space of k-dimensional arrays to those which are computed by a specific function. Since the mTable is a datastructure that should represent the minimum proportions required for a specific dice roll, we have to define $\texttt{constructMTable}$ such that it represents this process. Otherwise we could think of various ways to define processes to compute possible mTables.
\begin{definition}[constructMTable]
	\label{def:05:construct-mtable-single-test}
	For $p\in [0,1], k \in \mathbb{N}, \alpha \in [0,1]$ we define a function to construct a mTable from input parameters $p, k, \alpha$ according to \cite{yang2016measuring}.

	\noindent$\texttt{constructMTable} : \\ (0,1) \times \mathbb{N} \times [0,1] \longrightarrow \lbrace (m(1) ,\ldots, m(k)): m(i) = F^{-1}(i,p,\alpha), \, i = \{1,\ldots,k\}\rbrace$ \\
	with $\texttt{constructMTable}(p,k,\alpha)=\text{mTable}_{p,k,\alpha}$.
\end{definition}
\begin{lemma}
	\label{lemma:05:legal-valid-mtable}
	If a mTable is legal, it is also valid.
\end{lemma}
\noindent Lemma \ref{lemma:05:legal-valid-mtable} follows directly by construction.
Now we need a discrete partition of the continuous $\alpha$ space, that is a discrete measure that corresponds to exactly one legal mTable for a given set of parameters $k,p,\alpha$.
We call this measure the \emph{mass} of a mTable.
%
%
%
%
\begin{definition}[Mass of a mTable]
	\label{def:05:Mass of a MTable}
	For $\text{mTable}_{p,k,\alpha}=(m(1) , m(2) , \ldots , m(k))$ we call\\
	$L_1(\text{mTable}_{p,k,\alpha})=\sum_{i=1}^k m(i)$ the \textit{mass of $\text{mTable}_{p,\alpha,k}$}.
\end{definition}
In the following we relate the continuous $\alpha$-space to the discrete mass of a mTable.
\begin{lemma}
	\label{lemma:05:non-decreasing-with-alpha-mtable}
	Every $\text{mTable}_{p,k,\alpha}=\texttt{constructMTable}(p,k,\alpha)=(m(1) , m(2) , \ldots , m(k))$ is non-decreasing with $\alpha$. This means that
	$\texttt{constructMTable}(p,k,\alpha - \epsilon) = (m(1)' , m(2)' , \ldots , m(k)')$ will result in $m(i)' \leq m(i)$ for $i=1,\ldots,k$ and $\epsilon > 0$.
\end{lemma}
\begin{proof}\label{proof:05:non-decreasing-with-alpha-mtable}
Every entry $m(i)$ for $i=1,\ldots,k$ is computed by line 3 of Algorithm~\ref{alg:constructMTable}, i.e. every entry is the inverse binomial cdf $F^{-1}(i,p,\alphaadj)$. In other words $m(i)$ is the smallest integer such that
\begin{equation}
\alphaadj \leq \sum_{j=0}^{m(i)}\binom{i}{j}p^i (1-p)^{i-j} = F^{-1}(i,p,\alphaadj)
\end{equation}
The following equivalence holds:
\begin{equation}\label{eq:smaller m}
\alphaadj - \epsilon \leq \sum_{j=0}^{m'(i)}\binom{i}{j}p^i (1-p)^{i-j} = F^{-1}(i,p,\alphaadj-\epsilon)
	\Leftrightarrow \alphaadj \leq \sum_{j=0}^{m'(i)}\binom{i}{j}p^i (1-p)^{i-j} + \epsilon
\end{equation}
Now suppose that $m'(i) > m(i)$ which contradicts lemma \ref{lemma:05:non-decreasing-with-alpha-mtable}.
Then it is that
\[m'(i)>m(i) \Rightarrow \sum_{j=0}^{m'(i)}\binom{i}{j}p^i (1-p)^{i-j} \geq \sum_{j=0}^{m(i)}\binom{i}{j}p^i (1-p)^{i-j}\]
because $\binom{i}{j}p^i (1-p)^{i-j} \geq 0 \; \forall i,j,p$. It follows for $\epsilon >0$ that
\begin{equation}
\alphaadj \leq \sum_{j=0}^{m(i)}\binom{i}{j}p^i (1-p)^{i-j} + \epsilon \leq \sum_{j=0}^{m'(i)}\binom{i}{j}p^i (1-p)^{i-j} + \epsilon
\end{equation}
But then $m'(i)\neq F^{-1}(i,p,\alphaadj-\epsilon)$ because $m(i)$ would be the smaller integer that satisfies Equation~\ref{eq:smaller m}. Thus it has to be that $m'(i)\leq m(i)$.
\end{proof}
\noindent This property shows that, if we reduce $\alpha$ in our binary search, the mass of the corresponding mTable is also reduced or stays the same.
It very usefully implies a criterion to stop the binary search: namely we stop the calculation when the mass of the mTable at the left search boundary equals the right search boundary.
Of course this only works if there exists exactly one legal mTable for each mass, which we proof in the following.
\begin{theorem}
	\label{theorem:05:mtable-mass-injection}
	For fix $p,k$ there exists exactly one legal mTable for each mass $L_1\in \lbrace 1,\ldots,k \rbrace$.
\end{theorem}
\begin{proof}
	\label{proof:05:mtable-mass-injection}
	We prove this by contradiction: Let $MT_{p,k,\alpha_1}$ and $MT'_{p,k, \alpha_2}$ be two different mTables with $L_1(MT_{p,k,\alpha_1 }) = L_1(MT'_{p,k,\alpha_2})$.\\
	If both are legal then it applies that $\texttt{constructMTable}(p,\alpha_1 ,k)=MT_{p,\alpha_1 ,k}$
	and \\ $\texttt{constructMTable}(p,\alpha_2 ,k)=MT'_{p,\alpha_2 ,k}$.
	Because $MT_{p,k,\alpha_1} \neq MT'_{p,k,\alpha_2}$, without loss of generality entries $m(i), m(i)' , m(j) , m(j)'$ exist in each table, such that $|m(i) - m(i)'| = |m(j) - m(j)'|$ while at the same time $m(i) > m(i)'$ , $m(j) < m(j)'$ for $i<j, i,j \in \lbrace 1, \ldots , k \rbrace$.
	(Think of it as the two entries in each table "evening out", such that both tables have the same mass.)\\
	If $\alpha_1 > \alpha_2$, then the statement $m(j)' > m(j)$ violates Lemma~\ref{lemma:05:non-decreasing-with-alpha-mtable}.
	If $\alpha_2 > \alpha_1$, then the statement $m(i) > m(i)'$ also violates Lemma~\ref{lemma:05:non-decreasing-with-alpha-mtable}.
	The only possibility left is hence that $\alpha_1 = \alpha_2$, which contradicts $MT \neq MT'$, as both are created using function \texttt{constructMTable}.
\end{proof}
With these mathematical properties we can perform a binary search on the continuous $\alpha$-space to find the corrected significance level $\alphaadj$.
This corrected significance is used to compute a final mTable with an overall failure probability $\failprob = \alpha$.
\begin{algorithm}[t!]
	\caption{Algorithm \algoBinomBinary calculates the corrected significance level $\alpha_c$ and the mTable $m_{\alpha_c , k, p)}$ with an overall probability $\alpha$ of rejecting a fair ranking.}
	\label{alg:05:binarySearch} 
	\footnotesize
	\AlgInput{$k$, the size of the ranking to produce; $p$, the expected proportion of protected elements; $\alpha$, the desired significance level.}
	\AlgOutput{$\alphaadj$ the adjusted significance level; \texttt{m\_{adjusted}} the adjusted mTable}
	\AlgComment{initialize all needed variables}
	\texttt{aMin $\leftarrow$ 0};
	\texttt{aMax $\leftarrow \alpha$ };
	\texttt{aMid} $\leftarrow \frac{(\texttt{aMin + aMax})}{2}$ \\
	\texttt{m\_min} $\leftarrow$ \texttt{constructMTable(k,p,aMin)}; 
	\texttt{m\_max} $\leftarrow$ \texttt{constructMTable(k,p,aMax)}; \\
	\texttt{m\_mid} $\leftarrow$ \texttt{constructMTable(k,p,aMid)} \\
	\texttt{maxMass} $\leftarrow$ \texttt{m\_max.getMass()};
	\texttt{minMass} $\leftarrow$ \texttt{m\_min.getMass()};
	\texttt{midMass} $\leftarrow$ \texttt{m\_mid.getMass()}\\
	
	\While{\texttt{minMass} $<$ \texttt{maxMass} AND \texttt{m\_mid.getFailProb()} $\neq \alpha$ }{
		\If{\texttt{m\_mid.getFailProb()} $< \alpha$}{
			\texttt{aMin} $\leftarrow$ \texttt{aMid}
			\texttt{m\_min} $\leftarrow$ \texttt{constructMTable(k,p,aMin)} \\
		}
		\If{\texttt{m\_mid.getFailProb()} $> \alpha$}{
			\texttt{aMax} $\leftarrow$ \texttt{aMid} \\
			\texttt{m\_max} $\leftarrow$ \texttt{constructMTable(k,p,aMax)} \\
		}
		\texttt{aMid} $\leftarrow \frac{(\texttt{aMin + aMax})}{2}$ \\
		\tcp{stop criteria if midMass equals maxMass or midMass equals minMass}
		\If{\texttt{maxMass - minMass == 1}}{
			\texttt{minDiff} $\leftarrow |$\texttt{m\_min.getFailProb() - }$\alpha|$ \\
			\texttt{maxDiff} $\leftarrow |$\texttt{m\_max.getFailProb() - }$\alpha|$ \\
			\tcp{return the $\alpha_c$ which has the lowest difference from the desired significance}			
			\If{\texttt{minDiff} $<$ \texttt{maxDiff}}{
				\Return{\texttt{aMin, m\_min}}
			}
			\Else{
				\Return{\texttt{aMax, m\_max}}
			}
		}
		\tcp{stop criteria if midMaxx is exactly the mass between minMass and maxMass}
		\If{\texttt{maxMass - midMass == 1} AND \texttt{midMass - minMass == 1}}{
			\texttt{minDiff} $\leftarrow |$\texttt{m\_min.getFailProb() - }$\alpha|$ \\
			\texttt{maxDiff} $\leftarrow |$\texttt{m\_max.getFailProb() - }$\alpha|$ \\
			\texttt{midDiff} $\leftarrow |$\texttt{m\_mid.getFailProb() - }$\alpha|$ \\
			\tcp{return the $\alpha_c$ which has the lowest difference from the desired significance}			
			\If{\texttt{midDiff} $\leq$ \texttt{maxDiff} AND \texttt{midDiff} $\leq$ \texttt{minDiff}}{
				\Return{\texttt{aMid, m\_mid}}
			}
			\If{\texttt{minDiff} $\leq$ \texttt{midDiff} AND \texttt{minDiff} $\leq$ \texttt{maxDiff}}{
				\Return{\texttt{aMin, m\_min}}
			}
			\Else{
				\Return{\texttt{aMax, m\_max}}
			}
		}
	}
	\Return{\texttt{aMid, m\_mid}}
\end{algorithm}

\section{Complexity Analysis}

In order to estimate the complexity of the whole procedure (and hence understand its computational feasibility), we need to know how many mTables exist for fix $k$ and $p$.
%
This is the number of non-decreasing sequences of integers that end with a number smaller or equal to $k$ and have length~ $k$.
\begin{theorem}
	\label{theorem:05:number-of-mtables}
	The number of legal mTables for $k,p$ is less or equal to $\frac{k(k-1)}{2}$ .
\end{theorem}

\begin{proof}
	\label{proof:05:number-of-mtables}
	Given the proof of Theorem~\ref{theorem:05:mtable-mass-injection} we can count the number of legal mTables for fix $p,k$ as follows: The maximum mass of a legal mTable of length $k$ is by construction
	$L_1 ((m(1) = 1,m(2) = 2, \ldots, m(k) = k)) = \sum_{i=1}^k m(i) = \frac{k(k-1)}{2}$.
	Following definition~\ref{def:05:valid-mtable} the 	    entry $m(1)$ is the smallest entry or equal to all other entries.
	Furthermore, because this mTable is legal and following Definition~\ref{def:05:legal-mtable}, the mTable is a result of Algorithm~\ref{alg:constructMTable}.
	Thus $m(1) = F^{-1}(1,p,\alpha) \in \lbrace 0,1\rbrace$. In other words, $m(1)$ can only be $0$ or $1$.

	In turn $m(2)$ can only be $2$, if $m(1)$ was $1$ (otherwise $m(2)<2$).
	Accordingly, the minimum mass of a legal mTable is $L_1((m(1),m(2), \ldots, m(k))) = 0$, if all $m(i)=0$.
	Following Lemma~\ref{lemma:05:non-decreasing-with-alpha-mtable}, we can create mTables with higher masses by increasing $\alpha$.
	Furthermore, following Theorem~\ref{theorem:05:mtable-mass-injection}, if a legal mTable exists for a given mass and parameters $k,p$, then this is the only existing legal mTable with that mass.
	We know that there is a theoretical minimum mass for legal mTables ($L_1 =0$) which would occur, for example, if we set $\alpha = 0$ assuming $p<1$.
	There also exists a theoretical maximum mass for legal mTables which is $\frac{k(k-1)}{2}$.
	At best, we can achieve every possible mass between those two to extremes.
	It follows that there are at most $\frac{k(k-1)}{2}$ masses for legal mTables of size $k$ for a fix~$p$.
\end{proof}
\subsubsection{\algoMtable complexity}\label{subsubsec:construct-mtable-complexity}
\algoMtable computes the inverse binomial cdf for $k$ positions of the ranking and stores each of the computed values in the MCDF Cache.
This leads to a time complexity of $\mathcal{O}(k) \cdot \mathcal{O}(F^{-1}(p,k,\alpha))$.
Assuming a constant time for the calculation of the binomial probability mass function, the time complexity of $F^{-1}(p,k,\alpha)$ for our implementation is $\mathcal{O}(i^2)$, where $i$ is the current position we calculate $F^{-1}$ for.
Note that the complexity of $F^{-1}$ depends on the desired accuracy of the computation.
The space complexity is $\mathcal{O}(k)$, if we do not store any intermediate results for future calculations.
\begin{table}[b!]
	\vspace{2mm}
	\scalebox{0.75}{
		\begin{tabular}{lll}
			\toprule
			\textbf{Algorithm} & \textbf{Time Complexity} & \textbf{Space Complexity}\\
			\midrule
			\rowcolor[HTML]{C0C0C0}
			\algoMtable & $\mathcal{O}(k) \cdot \mathcal{O}(F^{-1}(p,k,\alpha))$ & $\mathcal{O}(k)$ \\
			\algoRecursive & $\mathcal{O}(\texttt{\algoMtable}) + \mathcal{O}(\prod_{j=1}^{m(k)}b(j) \cdot O(\texttt{binomPDF}))$ & $\mathcal{O}(k)$ \\
			\rowcolor[HTML]{C0C0C0}
			\algoBinomBinary & $\mathcal{O}(\log{}k) \cdot (\mathcal{O}(\texttt{\algoRecursive}))$ & $\mathcal{O}(k)$  \\
			\bottomrule
		\end{tabular}
	}
	\caption{Time complexity for all algorithms for one protected group without pre-computed results.\label{tbl:time-space-binom}}
\end{table}

\subsubsection{\algoRecursive complexity}\label{subsubsec:success-prob-complexity}
The algorithm \algoRecursive has time complexity $\mathcal{O}(\prod_{j=1}^{m(k)}b(j) \cdot O(\texttt{binomPDF}))$ as explained in section \ref{subsubsec:adjustment-binomial}.
Before we compute the success probability, we have to calculate the corresponding mTable and blocks $b$ which adds $\mathcal{O}(k) \cdot \mathcal{O}(F^{-1}(p,k,\alpha))$ and $\mathcal{O}(k)$ to the time complexity of \algoRecursive.
Overall we get $\mathcal{O}(k) \cdot \mathcal{O}(F^{-1}(p,k,\alpha)) + \mathcal{O}(\prod_{j=1}^{m(k)}b(j) \cdot O(\texttt{binomPDF}))$.
For the sake of readability we will write $\mathcal{O}($\algoMtable$) + \mathcal{O}(\prod_{j=1}^{m(k)}b(j) \cdot O(\texttt{binomPDF}))$.
The space complexity is $\mathcal{O}(k)$ for the maximum number of blocks plus $\mathcal{O}(k)$ for the stored probabilities at each position.
\subsubsection{\algoBinomBinary complexity}\label{subsubsec:binom-binary-complexity}
A general binary search on a list of $n$ items has a time complexity of $\mathcal{O}(\log{}n)$. We showed with theorem \ref{theorem:05:number-of-mtables} that the list of mTables on which we will search binary, has a maximum size of $\frac{k(k-1)}{2}$.
Thus the binary search for $\alpha_c$ has a complexity of $\mathcal{O}(\log{}\frac{k(k-1)}{2}) = \mathcal{O}(\log{}k^2) = \mathcal{O}(\log{}k)$.
For each binary search step we need $\mathcal{O}(\mathcal{O}(k) \cdot \mathcal{O}(F^{-1}(p,k,\alpha)))$ to compute the new mTable, as well as $\mathcal{O}(\prod_{j=1}^{m(k)}b(j) \cdot O(\texttt{binomPDF}))$ for its fail probability.
Overall we get $\mathcal{O}(\log{}k) \cdot (\mathcal{O}(\prod_{j=1}^{m(k)}b(j) \cdot O(\texttt{binomPDF})) + \mathcal{O}(\mathcal{O}(k) \cdot \mathcal{O}(F^{-1}(p,k,\alpha))))$, which we will write as $\mathcal{O}(\log{}k) \cdot (\mathcal{O}(\texttt{\algoRecursive}))$.
The space complexity is $\mathcal{O}(k)$ since we only store the three mTables with their respective fail probability at a time.

\section{Implementation}
We implemented the presented correction and algorithms into \textsc{FairSearch}, an open-source API for fairness in ranked search results~\cite{zehlike2020fairsearch}. 
The code can be found at \url{https://github.com/fair-search}.



\bibliographystyle{ACM-Reference-Format}
\bibliography{main}
\end{document}